\documentclass[a4paper,11pt]{article}

\usepackage[nomath]{kpfonts}
\usepackage[utf8]{inputenc}
\usepackage[margin=3cm]{geometry}

\usepackage{xspace}
\usepackage{mathtools}
\usepackage{amsmath,amsfonts,amsthm,amssymb,color}
\usepackage{graphicx}
\usepackage{caption}
\usepackage{subcaption}
\usepackage{authblk}

\usepackage[nocompress,noadjust]{cite}
\usepackage[dvipsnames,usenames,cmyk]{xcolor}
\usepackage[colorlinks=true,urlcolor=blue,citecolor=green,linkcolor=blue,breaklinks,unicode]{hyperref}
\usepackage[nameinlink,capitalize,noabbrev]{cleveref}

\usepackage[textsize=footnotesize]{todonotes}

\newcommand{\alg}[1]{{\sc #1}}
\newcommand{\prob}[1]{\textsc{#1}\xspace}
\newcommand{\df}[1]{{\it #1}}
\newcommand{\etal}{\textit{et al\@.\xspace}}

\newtheorem{example}{Example}
\newtheorem{definition}{Definition}
\newtheorem{lemma}{Lemma}
\newtheorem{theorem}{Theorem}

\usepackage{tikz}
\usetikzlibrary{arrows.meta,bending,positioning, decorations.pathmorphing}  
\usetikzlibrary{calc}

\title{Minimum Coverage Instrumentation}
\author[1]{Li Chen}
\author[2]{Ellis Hoag}
\author[2]{Kyungwoo Lee}
\author[2,3]{Juli\'{a}n Mestre}
\author[2]{Sergey Pupyrev}
\date{}
\affil[1]{Department of Computer Science, Georgia Tech, USA.}
\affil[2]{Meta Platforms Inc., USA.}
\affil[3]{School of Computer Science, University of Sydney, Australia.}
\begin{document}

\maketitle

\begin{abstract}    
    Modern compilers leverage block coverage profile data to carry out downstream profile-guided 
    optimizations to improve the runtime performance and the size of a binary. 
    Given a control-flow graph $G=(V, E)$ of a function in the binary, where nodes in $V$ correspond to basic blocks (sequences of instructions that are always executed sequentially) and edges in $E$ represent jumps in the control flow, the goal is to know for each block $u \in V$ whether $u$ was executed during a session. To this end, extra instrumentation code that records when a block is executed needs to be added to the binary. This extra code creates a time and space overhead, which one would like to minimize as much as possible.
    
    Motivated by this application, we study the \prob{minimum coverage instrumentation} problem, where the goal is to find a minimum size subset of blocks to instrument such that the coverage of the remaining blocks in the graph can be inferred from the coverage status of the instrumented subset. Our main result is an algorithm to find an optimal instrumentation strategy and to carry out the inference in $O(|E|)$ time.    
    We also study variants of this basic problem in which we are interested in learning the coverage of edges instead of the nodes, or when we are only allowed to instrument edges instead of the nodes.
\end{abstract}

\section{Introduction}

Code profiling is an important tool in modern compilers that unlocks downstream analysis and optimizations of binaries based on their
run-time behavior. Arguably, the most commonly supported profiling primitive is frequency counts: Given a control-flow graph associated with a function, the compiler injects additional code to record how many times each node (representing a basic block of instructions) or each edge (representing jumps in the control flow) is executed~\cite{BL1994,HMPWY22}. This profile data can then be used to carry out profile-guided optimization of the binary to improve its run-time performance. For example, one can optimize the layout of basic blocks in memory to decrease the number of instruction cache misses incurred while fetching the code for execution and therefore improve the overall performance of the binary~\cite{PH90,NewellP20,MestrePS20}.
Another prominent use-case of frequency counts is reducing the size of mobile applications via improved function outlining~\cite{CLB21,LeeHT22}.

In this paper we focus on the computational problem of profiling block coverage. Unlike the frequency profiling in which
the goal is to count the frequency of every block, the coverage profiling asks whether a block has been executed during a session.
Coverage instrumentation is important for identifying gaps in program test design~\cite{Agrawal94} and is used to guide 
optimizations in modern mobile compilers~\cite{LeeHT22}. Given a (directed) control-flow graph $G=(V, E)$ of a function, we add extra 
instrumentation code that records when a block is executed.
While in principle one could use counts to infer coverage, this is not a practical approach as the overhead of instrumenting block frequencies is much higher than the overhead of block coverage\footnote{The overhead is higher both in terms of the binary size (a counter typically requires 4 bytes for an integer versus 1 byte for a boolean) as well as time (updating a counter typically requires two extra machine-level instructions to load and increment the count before its value is stored).}.
Another simple strategy for coverage instrumentation is to add a (boolean) counter at every block. 
This might however, incur an unnecessary overhead, as
not every block needs to be instrumented to determine the coverage status of every block in the function.
For example, to learn the coverage of a chain of blocks (all having in- and out-degrees of $1$),
it is sufficient to instrument the coverage of only one block in the chain.
Thus our goal is to minimize the overhead as much as possible. In other words, 
we want to find a minimum size subset of nodes to instrument such that it is always possible to reconstruct the coverage of all nodes in the graph from the coverage of instrumented ones. 

Our main result is an optimal algorithm for the problem 
(formally defined in \cref{sect:prob}) that finds the smallest set of blocks to instrument and carries out the inference in $O(|E|)$ time.
We also study a variant of this basic problem where we are interested in learning the coverage of edges instead of nodes, and another variant where we can instrument edges instead of nodes. For edge-coverage edge-instrumentation we are able to get an optimal algorithm and for vertex-coverage edge-instrumentation we develop an approximation algorithm.

\subsection{Related Work}

Profile-guided optimization is an essential step in modern compilers; we refer to \cite{BL1994,PH90,LeeHT22,HMPWY22} and 
references thereof for an overview of the field. A classical problem in the area is 
that of profiling binaries to compute frequency counts. The study of how many basic blocks need to be instrumented to compute frequency counts goes back to the 70's. Nahapetian~\cite{Nah1973} determined the necessary number of blocks to instrument via certain reduction rules of the control-flow graph. Knuth~and~Stevenson~\cite{KS1973,Knuth73} provide an alternative interpretation of the algorithm and a proof of optimality.
The latter is based on computing the minimum spanning tree in the graph and is a part of most modern instrumentation-based
profiling tools.

Ball and Larus~\cite{BL1994} define a hierarchy of frequency profiling problems.
These problems have two dimensions: what is profiled in the control-flow graph (that is, basic blocks or jumps) and where the instrumentation code is placed (blocks or jumps). They denote the problems as follows. The vertex profiling problem is to determine block frequencies and is denoted by $Vprof(\cdot)$, while the edge profiling problem is to determine jump frequencies and is denoted $Eprof(\cdot)$.
Given that one can place counters on blocks or jumps, there are two placement strategies, $Vcnt$ and $Ecnt$. 
As such there are four problems with known algorithmic results:
\begin{itemize}	
	\item $Eprof(Ecnt)$: Solved optimally by Knuth~\cite{Knuth73} using spanning trees.
	
	\item $Vprof(Vcnt)$: Solved optimally by Nahapetian and Knuth~\etal~\cite{Nah1973,KS1973} via a reduction to $Eprof(Ecnt)$.

	\item $Eprof(Vcnt)$: There exist instances where the edge counts cannot be uniquely determined from vertex counts~\cite{Pro1982}, so this problem does not admit an algorithm.
	
	\item $Vprof(Ecnt)$: Ball and Larus~\cite[Sect. 3.3]{BL1994} provide a characterization of when the set of edges is sufficient for determining all vertex frequencies. However, the complexity of the minimization problem remains open.	
\end{itemize}

Coverage profiling, on the other hand, has received much less attention in 
the literature and, to the best of our knowledge, has not been thoroughly studied from 
a theoretical point of view. And while it may be temping to think that one could
use an optimal solution for frequency count instrumentation as a basis for an optimal coverage instrumentation, the examples in \cref{app:examples} show that this is not a via approach since a feasible
solution for one problem need not be feasible for the other problem; indeed, the size of the optimal solution for these two problems can differ widely.

Agrawal~\cite{Agrawal94} considers several problems related to test coverage of control-flow graphs. The main focus of the work is to find a small subset of nodes $S$ such that any set of executions that covers $S$ also covers all other nodes, which is useful when designing tests. The paper also proposes an algorithm for finding a coverage instrumentation that runs in $O(|V||E|)$ time but does not provide a proof that the scheme has minimum size.

Tikir and Hollingsworth~\cite{TikirH05} propose using dynamic functions to reduce the profiling overhead. Their system periodically removes the instrumentation code that updates the coverage status of covered blocks since further executions do not provide additional information. As part of their system, they propose a linear-time heuristic for finding a coverage instrumentation scheme. However, their algorithm is not optimal and their approach is not technically feasible in some architectures.

Finally, on the practical side, various heuristics for coverage instrumentation have been implemented in compilers~\cite{asan} but the 
algorithms have no performance guarantees and may produce sub-optimal results.

\section{Problem Statement}
\label{sect:prob}

Let $G=(V, E)$ be a directed graph representing a control-flow graph (\df{CFG}). We assume $G$ has two distinct nodes $s$ and $t$, called the \df{entry node} and the \df{terminal node}, such that $\deg^{in}(s) = 0$ and $\deg^{out}(t) = 0$. Furthermore, 
every node in $G$ is reachable from $s$ and every node in $G$ can reach $t$ via a (directed) path.

An execution trace of $G$ is a collection of (not necessarily simple) $s$-$t$ paths. The (full) \df{coverage profile} associated with an execution trace is a truth assignment $C: V \rightarrow \{\top, \bot\}$ where $C(u) = \top$ if and only if $u$ is spanned by one of the paths in the execution trace. We let $\mathcal{C}$ be the collection of all coverage profiles induced by some execution trace of $G$. A \df{partial coverage profile} $C_S$ is the restriction of $C$ to a subset $S \subset V$.

A \df{coverage instrumentation scheme} consists of a set of nodes $S \subset V$ and an efficiently computable inference function $\Psi$ that, given a partial coverage profile defined on $S$, outputs a full coverage profile. We say that the coverage instrumentation scheme $(S, \Psi)$ is \df{valid} if for any valid coverage profile $C \in \mathcal{C}$, we have $\Psi(C_S) = C$. Finally, we define the \df{size} of the scheme to be $|S|$.

The \prob{minimum block coverage instrumentation} problem is to select a minimum size coverage instrumentation scheme $(S, \Psi)$. While our main metric for evaluating a scheme is its size, $|S|$, we also care, as a secondary metric, about the time complexity of both finding the scheme and of evaluating $\Psi$.

\begin{example}
    \label{ex:diamond}
    Consider the following toy instance with 
    \(V = \{ v_1, v_2, v_3, v_4 \}\) and\linebreak \(E = \{ (v_1, v_2), (v_2, v_4), (v_1, v_3), (v_3, v_4) \}\). 
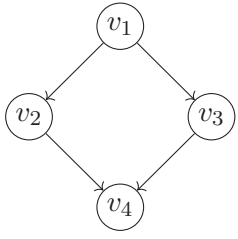
    \begin{center}
        \begin{tikzpicture}[
            block/.style={circle,draw=black,fill=white, inner sep = 2pt},
            jump/.style={->},
            scale=1.2
        ]
            \draw (0, 2) node[block] (v1) {$v_1$};
            \draw (-1, 1) node[block] (v2) {$v_2$};
            \draw (1, 1) node[block] (v3) {$v_3$};
            \draw (0, 0) node[block] (v4) {$v_4$};
            \draw (v1) edge[jump] (v2);
            \draw (v1) edge[jump] (v3);
            \draw (v2) edge[jump] (v4);
            \draw (v3) edge[jump] (v4);
        \end{tikzpicture}            
    \end{center}    

    Notice that in this case 
    \[\mathcal{C} = \{ \emptyset, \{v_1, v_2, v_4\}, \{v_1, v_3, v_4\}, \{v_1, v_2, v_3, v_4 \} \}.\]

    It is easy to see that the optimal solution is $S=\{v_2, v_3\}$ and for any partial coverage $C_S$ induced by $C \in \mathcal{C}$ we have
    \[ \Psi(C_S) = 
    \begin{cases}
        C_S \cup \{v_1 \rightarrow \bot, v_4 \rightarrow \bot\} & \text{if } C_S = \{v_2 \rightarrow \bot, v_3 \rightarrow \bot\} \\
        C_S \cup \{v_1 \rightarrow \top, v_4 \rightarrow \top\}  & \text{o.w.}
    \end{cases}  \]    
\end{example}

\subsection{Our results}

While the main focus of this paper is on \prob{minimum block coverage instrumentation}, we also consider related variants of the main problem where we want to compute edge coverage and/or where we are allowed to instrument edges. Analogously to the hierarchy of Ball and Larus~\cite{BL1994}
for frequency profiling, 
this gives rise to four problems, which we denote by 
\prob{$X$-cov $Y$-instr} for $X, Y \in \{V, E\}$, where we want to compute coverage data of $X$ while instrumenting a subset of $Y$.

Our results for these variants are as follows:
\begin{enumerate}
    \item \prob{$V$-cov $V$-instr}: In \cref{sec:main} we give an optimal linear-time algorithm.
    \item \prob{$E$-cov $E$-instr}: In \cref{sec:reduction} we show a reduction to \prob{$V$-cov $V$-instr}.
    \item \prob{$E$-cov $V$-instr}: In \cref{sec:impossible} we show that it is not possible in general to infer edge coverage from vertex coverage data.
    \item \prob{$V$-cov $E$-instr}: In \cref{sec:approximation} we give a 2-approximation algorithm.
\end{enumerate}

\section{Algorithmic Framework}

\label{sec:framework}

Before we describe and prove the correctness of our algorithm, we need to develop some basic graph theoretic concepts. 

\begin{definition}
    For any vertex $u \in V$ we let
    \begin{itemize}
        \item $A(u)$ be the set of nodes that can be reached from $s$ while avoiding $u$, and
        \item $B(u)$ be the set of nodes that can reach $t$ while avoiding $u$.
    \end{itemize}     
\end{definition}

Observe that $A(s) = \emptyset$, $u \notin A(u) \cup B(u)$ for any $u \in V$, and $B(t) = \emptyset$. Note that we can compute $A(u)$ and $B(u)$ in $O(|E|)$ time for a fixed $u \in V$ using a modified BFS or DFS search. We state a simple observation about these sets that will be useful later on.

\begin{lemma}
    \label{lem:AorA-BorB}
    For any two vertices $u, v \in V$ we have:
    \begin{itemize}
        \item $v \in A(u)$ or $u \in A(v)$, and
        \item $v \in B(u)$ or $v \in B(v)$.
    \end{itemize} 
\end{lemma}

\begin{proof}
    Let $P$ be a simple path $s$-$u$. If $v \notin P$ then $u \in A(v)$. Otherwise, trim the path to get an $s$-$v$ path that avoids $u$, which shows that $v \in A(u)$. Thus, proving the first statement. The second statement is proved analogously starting with a $u$-$t$ path.
\end{proof}

\subsection{Ambiguous nodes}

\begin{definition}
    We say a node $u$ is \df{ambiguous} if
    \begin{itemize}
        \item $ \exists \, x \in N^{in}(u) \cap (A(u) \cap B(u))$, and
        \item $ \exists \, y \in N^{out}(u) \cap (A(u) \cap B(u))$.
    \end{itemize}
\end{definition}

\begin{lemma}
    \label{lem:need-ambiguous}
    Let $(S, \Psi)$ be a valid scheme then every ambiguous node $u \in V$ must belong to $S$.
\end{lemma}

\begin{proof}
    Let $x \in N^{in}(u) \cap (A(u) \cap B(u))$; that is, there exist $s$-$x$ and $x$-$t$ paths avoiding $u$. We can concatenate both paths to form an $s$-$t$ path $P_x$ going through $x$ that avoids $u$. The same line of reasoning applied to $y \in N^{out}(u) \cap (A(u) \cap B(u))$ yields another $s$-$T$ path $P_y$ going through $y$ avoiding $u$. Finally, consider concatenating the $s$-$x$ path, followed by $u$, followed by the $y$-$T$ path, and call $P_u$ the resulting $s$-$t$ path. \cref{fig:ambiguous} shows an example instance of what these paths may look like.

    \begin{figure}
        \centering
        \begin{tikzpicture}[
            block/.style={circle,draw=black,fill=white, inner sep = 2pt},
            jump/.style={->},
            long_path/.style={
                decorate,
                decoration={
                    snake,
                    amplitude=.4mm,
                    segment length=2mm,
                    post length=2.5mm,
                    pre length=2.5mm
                },
                ->
            },
            scale=1.2
        ]
            \draw (0, 0) node[block] (s) {$s$};
            \draw (6, 0) node[block] (t) {$t$};
            \draw (2, 0) node[block] (x) {$x$};
            \draw (3, 0) node[block] (u) {$u$};
            \draw (4, 0) node[block] (y) {$y$};
            \draw (x) edge[jump] (u);
            \draw (u) edge[jump] (y);

            \draw[red] (s) edge[long_path] (x);
            \draw[red] (x) edge[long_path, bend right=45] node[below,yshift=-2pt] {$P_x$} (t);

            \draw[blue] (s) edge[long_path, bend left=45] node[above,yshift=2pt] {$P_y$} (y);
            \draw[blue] (y) edge[long_path] (t);

            \draw[black,thick,dotted, rounded corners] ($(s.north west)+(-0.2,0.2)$)  rectangle ($(t.south east)+(0.2,-0.2)$);

            \draw (-0.75,0.5) node {$P_u$};
        \end{tikzpicture}

        \caption{
            \label{fig:ambiguous}
            An example of paths $P_x$ (in red), $P_y$ (in orange), and $P_u$ (in dotted black).
        }
    \end{figure}
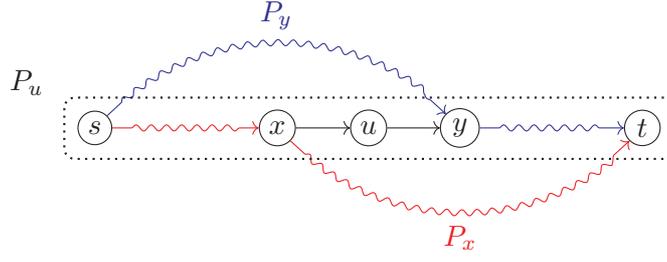
    
    Let $D$ and $D'$ be the coverage profiles associated with $\{ P_x, P_y \}$ and $\{P_x, P_y, P_u \}$ respectively. Notice that $D(v) = D'(v)$ for $v \neq u$ and $D(u) \neq D'(u)$. Therefore, even if we knew the coverage of every nodes in $V - u$, it is not possible to differentiate between $D$ and $D'$ unless $u \in S$.
\end{proof}

It is worth noting that while every ambiguous node must be part of a valid scheme, these nodes by themselves may not form a valid scheme, in which case additional nodes are needed.

\begin{example}
    Consider the following example with 
    \(V = \{ v_1, v_2, v_3 \}\) and\linebreak \(E = \{ (v_1, v_2), (v_2, v_3), (v_2, v_3) \}\). 

    \begin{center}
        \begin{tikzpicture}[
            block/.style={circle,draw=black,fill=white, inner sep = 2pt},
            jump/.style={->},
            scale=1.2
        ]
            \draw (0, 2) node[block] (v1) {$v_1$};
            \draw (-1, 1) node[block] (v2) {$v_2$};
            \draw (0, 0) node[block] (v3) {$v_3$};

            \draw (v1) edge[jump] (v2);
            \draw (v1) edge[jump] (v3);
            \draw (v2) edge[jump] (v3);
        \end{tikzpicture}
    \end{center}

    Notice that in this case 
    \[\mathcal{C} = \{ \emptyset, \{v_1, v_3\}, \{v_1, v_2, v_3\}\}.\]
    and $v_2$ is the only ambiguous node. However, $\{v_2\}$ is not enough to distinguish between profiles $\{ v_1 \rightarrow \bot, v_2 \rightarrow \bot, v_3 \rightarrow \bot\}$ and $\{ v_1 \rightarrow \top, v_2 \rightarrow \bot, v_3 \rightarrow \top\}$. Thus another node needs to be instrumented (either $v_1$ or $v_3$).
\end{example}

\subsection{Forward and backward inference}

At the heart of our method is the concept of the forward and backward inference graphs, which we define next.

\begin{definition}
    We say that a node $u \in V$ is \df{forward inferable} if $N^{out}(u) \setminus A(u) \neq \emptyset$ and $N^{out}(u) \cap (A(u) \cap B(u)) = \emptyset$. And we define the forward inference graph $(V, F)$ where $(u, v) \in F$ if $u$ is forward inferable and $v \in N^{out}(v) \setminus A(u)$.
\end{definition}

\begin{definition}
    We say that a node $u \in V$ is \df{backward inferable} if $N^{in}(u) \setminus B(u) \neq \emptyset$ and  $N^{in}(u) \cap (A(u) \cap B(u)) = \emptyset$. And we define the backward inference graph $(V, D)$ where $(u, v) \in D$ if $u$ is backward inferable and $v \in N^{in}(u) \setminus B(u)$.
\end{definition}

Notice how the edges in the backward inference graph reverse the direction of the input graph edges that are used to make the inference. This is because we want the inference graphs to capture the precedence constraints needed to make the inferences.

\begin{lemma}
    \label{lem:one-step}
    Suppose that $u$ is forward (backward) inferable, and let $X$ be the set of successors of $u$ in the forward (backward) inference graph. Let $C$ be a valid coverage profile, then 
    \[ \bigvee_{u \in X} C(u) \equiv C(v). \]
\end{lemma}

\begin{proof}
    We prove the forward inference case only as the backward inference case is analogous. Let $P$ be an $s$-$t$ path going through a vertex $v \in X$.  Since $X$ is non-empty, we know that such a path exists. Since $X \cap A(u) = \emptyset$, it follows that $u \in P$. Therefore, if $C(v) = \top$ then $C(u) = \top$.
    
    Now let $Q$ be an $s$-$t$ path going through $u$ and let $v$ be the vertex right after the last occurrence of $u$ in $Q$. This means that $v \in N^{out}(u) \cap B(u)$. Since $u$ in inferable, it follows that $v \notin A(u)$, Therefore, if $C(u) = \top$ then $\exists\, v \in X: C(v) = \top$.
\end{proof}

An \emph{inference scheme} is a partition $(\alpha, \phi, \beta)$ of $V$ into three parts where $\alpha$ is the set of instrumented nodes, $\phi$ is the set of forward inferable nodes, and $\beta$ is the set of backward inferable nodes.

The inference graph associated with an inference scheme $(\alpha, \phi, \beta)$ is the directed graph $H$ where for each $u \in \phi$, $\delta^{out}_H(u)$ is the set of forward inference edges out of $u$, and for each $u \in \beta$, $\delta^{out}_H(u)$ is the set of backward inference edges out of $u$. We say the scheme is valid if its associated inference graph $H$ is acyclic. 

Finally, we associate with an inference scheme $(\alpha, \phi, \beta)$ a coverage instrumentation scheme $(S, \Psi)$ where $S = \alpha$ and $\Psi(C_S)$ is the result of starting from the partial coverage profile $C_S$ and and iteratively applying \cref{lem:one-step} to those nodes in $\phi \cup \beta$ in inverse topological order $v_1, v_2, \ldots, v_n$ in $H$ (edges go from right to left).

\begin{lemma}
    \label{lem:valid-inference}
    Given a valid inference scheme $(\alpha, \phi, \beta)$, its associated coverage instrumentation scheme $(S, \Psi)$ is also valid. Furthermore, for any coverage profile $C$, the function $\Psi(C_S)$ can be evaluated in $O(|E|)$ time given the inference graph $H$.
\end{lemma}

\begin{proof}
    Let $C^{(i)}$ be the partial coverage profile resulting in taking $C^{(i-1)}$ and adding the result of processing vertex $v_i$. Namely, if $v_i \in \alpha$ then $C^{(i)}(v_i) = C_S(v_i)$, and if $v_i \in \phi \cup \beta$ then $C^{(i)}(v_i)$ is set using \cref{lem:one-step}. The correctness rests on the the fact that we can always apply \cref{lem:one-step} at each step because all the out going neighbors of $v_i$ in the inference graph upon which $v_i$ depends for its inference have been already processed earlier because the order of processing is inverse topological order.

    To see the claim about the time complexity of evaluating $\Psi$, we note that computing the needed topological order can be done in $O(|E|)$ time and that once we have that the iterative process also runs in linear time.
\end{proof}

Our approach is to show that there always exists an inference scheme $(\alpha, \phi, \beta)$ that induces an optimal coverage instrumentation scheme. And that such scheme can be computed efficiently in $O(|E|)$ time.

\section{Optimal $V$-coverage $V$-instrumentation}
\label{sec:main}

In this section we design an algorithm to compute an optimal coverage instrumentation scheme. For now, we are only concerned about its correctness and leave time complexity considerations for later. To this end, we study the cycle structure of inference graphs, which we will leverage to design our algorithm.

\begin{lemma}
    \label{lem:acyclic}
    The forward (backward) inference graph is acyclic.
\end{lemma}

\begin{proof}
    We give the proof for the forward inference graph as the proof for backward inference is analogous.

    Suppose, for the sake of contradiction, that the graph has a cycle $v_1, v_2, \ldots, v_k$. For some $i$, there must exists an $s$-$v_i$ path that avoids the rest of the cycle. This means that $v_i \in A(v_j)$ for all $j\neq i$. But since $(v_{i-1}, v_i)$ is a forward inference edge, this means that $v_i \in N^{out}(v_i) \setminus A(v_{i-1})$, which implies $v_i \in A(v_{i-1})$. Contradicting our assumption.
\end{proof}

\begin{lemma}
    \label{lem:short-cycles}
    Let $H$ be the union of the forward and backward inference graphs. Then every simple cycle in $H$ has at most two nodes.
\end{lemma}

\begin{proof}
    First, we note that because of \cref{lem:acyclic}, it must be the case that the cycle uses edges from both the forward and the backward inference graphs. That is, for some even $\ell$, there are nodes $u_0, u_1, \ldots, u_{\ell - 1} \in C$ such that the segment of the cycle from $u_i$ to $u_{i+1}$ is made up of forward edges if $i$ is even and backward edges if $i$ is odd; in this context, we use the notation $u_\ell = u_0$.

    First, let us consider the case where we alternate from forward to backward more than once; that is, $\ell > 2$. Without loss of generality suppose that there exists a $s$-$u_0$ path avoiding all other $u_2, u_4,\ldots$ (we can always relabel the nodes so that is the case). In particular, this means that there exists a path $P$ from $s$ to $u_1$ via $u_0$ that avoids $u_2$. This in turn means that there exists a backward inference edge $(x, y)$ along the segment of the cycle from $u_1$ to $u_2$ such that $x \in A(y)$ which contradicts the definition of backward inference edge.

    Now, let us consider the case where we alternate only once; that is, $\ell = 2$. Suppose that one of the two segment, say $u_0$-$u_1$ has two or more edges. Then we take a path from $s$ to $u$ (which must avoid at least the node ahead of $u_0$ in the cycle) and then follow the edges in the input graph inducing the reverse edges in the segment $u_1$-$u_0$. It follows that there must exists a forward inference edge $(x, y)$ along the segment $u_0$-$u_1$ such that $y \in A(x)$, which contradicts the definition of forward edge. The case when the long segment is made up of backward inference edges is handled in the similar fashion.

    The only case that remains to consider is when alternate only once and the segments $u_0$-$u_1$ and $u_1$-$u_0$ consist of a single edge. In this case the cycle consists of only two nodes, as prescribed by the lemma statement.
\end{proof}

A simple consequence of \cref{lem:short-cycles} is that a connected component in the union of the forward and backward inference graphs has a tree-like structure where every edge in the tree induces a pair of anti-parallel edges in the connected component. Our ultimate goal is to select an inference scheme that breaks these cycles by judiciously choosing to do either forward or backward inferences or by instrumenting additional nodes. Before we can do that, we need to better understand the structure of these connected components.

\begin{lemma}
    \label{lem:path}
    Let $H$ be the union of the forward and backward inference graphs and $C$ be a connected component in $H$. Then $C$ induces a directed path in the forward inference graph, and the reverse paths in the backward inference graph. There are no further edges in $H[C]$.
\end{lemma}

\begin{proof}
    As already explained connected components are made up of a collection of pairs of anti-parallel edges connected forming a tree-like structure. If the structure is not a tree then there must exists a node that has two incoming forward edges or two outgoing forward edges. Let us consider the former case; namely, there exists two forward edges $(x, u)$ and $(y, u)$. From \cref{lem:AorA-BorB} we know that $x \in A(y)$ or $y \in A(x)$; suppose without loss of generality the former is true. That is, there exists an $s$-$x$ path that avoids $y$, which we can extend by appending the edge $(x, u)$ thus showing that $u \in A(y)$. But this contradict the fact that $(y, u)$ is a forward inference edge. 
    
    Now let us consider the case where there exists two forward inference edges $(u, x)$ and $(u, y)$, or equivalently, that $(x, u)$ and $(y, u)$ are backward inference edges. From \cref{lem:AorA-BorB} we know that $x \in B(y)$ or $y \in B(y)$; suppose without loss of generality the former is true. That is, there exists a $x$-$t$ path that avoids $y$, which we can extend by pre-pending the edge $(u, x)$ and so $u \in B(y)$. But this contradicts the fact that $(y, u)$ is a backward inference edge. 

    Therefore, the tree-like structure of $C$ must be a path and must consist of one forward inference path and the reversed backward inference path. Otherwise, we are back the cases we just ruled out.
\end{proof}

Everything is in place to state our algorithm \alg{optimal-instrumentation}. Remember our goal is to construct an inference scheme $(\alpha, \phi, \beta)$. As a first step we identify all ambiguous nodes in $G$ and add them to $\alpha$. Then build the union of the forward and backward inference graphs and compute its connected components. Each connected component $C$ consists of a forward inference path $v_1, v_2, \ldots, v_k$ and a backward inference path $v_k, v_{k-1}, \ldots, v_1$. For a trivial component where $k=1$, if $v_1$ is ambiguous, there is nothing to do; otherwise, we add it to $\phi$ or $\beta$ depending on whether if forward or backward inferable (if both options are possible, pick one arbitrarily.) For non-trivial components where $k > 1$, if $v_1$ happens to backward inferable (from nodes outside the component) then add $C$ to $\beta$. Otherwise, if $v_k$ happens to be be forward inferable (from nodes outside the component) then we add $C$ to $\phi$. Finally, if neither $v_1$ is backward inferable or $v_k$ is forward inferable, we add $v_1$ to $\alpha$ and $\{v_2, \ldots, v_{k}\}$ to $\beta$. The output of the algorithm is the coverage instrumentation scheme $(S, \Psi)$ associated with $(\alpha, \phi, \beta)$.

\begin{theorem}
    \label{thm:optimal}
    \alg{optimal-instrumentation} returns a valid $V$-coverage $V-$instrumentation scheme that has minimum size.
\end{theorem}

\begin{proof}
    First we show that $(\alpha, \phi, \beta)$ is a valid inference scheme.     This boils down to arguing that the inference graph associated with the inference scheme is acyclic. Because of the special structure of the connected components in the union of the forward and backward inference graphs, and they way each component is dealt with, we are guaranteed that no cycles are present in the inference graph. Therefore, the inference scheme is valid. \cref{lem:valid-inference} guarantees that the associated coverage instrumentation scheme is also valid.

    Finally, we argue that the scheme has minimum size. In \cref{lem:need-ambiguous} we already argued that any valid instrumentation must use all ambiguous nodes. Let $C$ be a connected component where our solution instruments the head of the component. Let us argue that a coverage instrumentation scheme that does not instrument a single node in $C$ must be invalid. The argument is similar but applied to the first node in $C$ (that is not backward inferable) and the last node in $C$ (that is not forward inferable). 
    
    Let $x \in N^{in}(v_1) \cap (A(v_1) \cap B(v_1))$; that is, there exist $s$-$x$ and $x$-$t$ paths avoiding $v_1$. We can concatenate both paths to form an $s$-$t$ path $P_x$ going through $x$ that avoids $v_1$. This line of reasoning applied to $y \in N^{out}(v_k) \cap (A(v_k) \cap B(v_k))$ yields another $s$-$t$ path $P_y$ going through $y$ avoiding $v_k$. It's not hard to see that both paths avoid $C$ altogether. Finally, consider concatenating the $s$-$x$ path, followed by $C$, followed by the $y$-$t$ path, and call $P_C$ the resulting $s$-$t$ path. \cref{fig:connected-component} shows an example instance of what these paths may look like.

    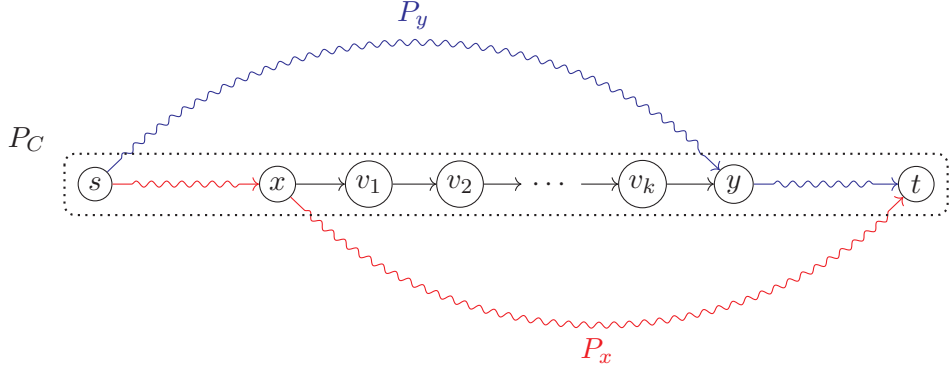
\begin{figure}
        \centering

        \centering
        \begin{tikzpicture}[
            block/.style={circle,draw=black,fill=white, inner sep = 2pt},
            jump/.style={->},
            long_path/.style={
                decorate,
                decoration={
                    snake,
                    amplitude=.4mm,
                    segment length=2mm,
                    post length=2.5mm,
                    pre length=2.5mm
                },
                ->
            },
            scale=1.2
        ]
            \path (0, 0) node[block] (s) {$s$}
             -- ++(2, 0) node[block] (x) {$x$}
             -- ++(1, 0) node[block] (v1) {$v_1$}
             -- ++(1, 0) node[block] (v2) {$v_2$}
             -- ++(1, 0) node (dots) {$\cdots$}
             -- ++(1, 0) node[block] (vk) {$v_k$}
             -- ++(1, 0) node[block] (y) {$y$}
             -- ++(2, 0) node[block] (t) {$t$};

            \draw (x) edge[jump] (v1);
            \draw (v1) edge[jump] (v2);
            \draw (v2) edge[jump] (dots);
            \draw (dots) edge[jump] (vk);
            \draw (vk) edge[jump] (y);

            \draw[red] (s) edge[long_path] (x);
            \draw[red] (x) edge[long_path, bend right=45] node[below,yshift=-2pt] {$P_x$} (t);

            \draw[blue] (s) edge[long_path, bend left=45] node[above,yshift=2pt] {$P_y$} (y);
            \draw[blue] (y) edge[long_path] (t);

            \draw[black,thick,dotted, rounded corners] ($(s.north west)+(-0.2,0.2)$)  rectangle ($(t.south east)+(0.2,-0.2)$);

            \draw (-0.75,0.5) node {$P_C$};
        \end{tikzpicture}

        \caption{
            \label{fig:connected-component}
            An example showing the paths that certify the need to instrument at least one block inside a connected component $C=\{v_1, \ldots, v_k\}$: Paths $P_{x}$ (in red), $P_{y}$ (in blue), and $P_C$ (in dotted black).
        }
    \end{figure}
    
    Let $D$ and $D'$ be the coverage profiles associated with $\{ P_{x}, P_{y} \}$ and $\{P_{x}, P_{y}, P_{C} \}$ respectively. Notice that $D(v) = D'(v)$ for $v \notin C$ and $D(v) \neq D'(v)$ for all $v \in C$. Therefore, even if we knew the coverage of every nodes in $V \setminus C$, it is not possible to differentiate between $D$ and $D'$ unless we instrument at least on node in $C$.
\end{proof}

\subsection{Time complexity}

Finally, we turn our attention to the time complexity of our algorithm.

\begin{theorem}
    \alg{optimal-instrumentation} can be implemented to run in $O(|E|)$ time.
\end{theorem}

\begin{proof}
    Note that if we could perform $O(1)$-time membership queries over the $A$ and $B$ sets, then the rest of the algorithm can be implemented to run in linear time: computing the forward and backward inference graphs, identifying the strongly connected components in the union of those graphs, and deciding on the instrumentation of those components can also be done in $O(|E|)$ time.

    We focus our attention on implementing membership queries of the sets $\{A(u) : u \in V \}$, as queries on the set $\{B(u) : u \in V\}$ can be implementing in a similar way by reversing the direction of the edges and using the terminal $t$ as the entry node.

    We use dominator trees~\cite{Prosser59}, which offer a compact representation of the dominance relation: a node $x$ dominates a node $y$ if and only if all paths from the entry node to $y$ go through $x$. The dominator tree is an out-branching rooted at the entry node such that if $x$ dominates $y$ if and only if $y$ is a descendant to $x$. While the naive algorithm for computing a dominator tree takes $O(|V|^2)$ time, more efficient $O(|E|)$-time algorithms exist~\cite{AlstrupHLT99,BuchsbaumGKRTW08,BuchsbaumKRW05,GeorgiadisT04}.

    Notice that $v \in A(u)$ if an only if $u$ does not dominate $v$. Thus, testing if $v \in A(u)$ can be translated of the query of whether $u$ is not an ancestor of $v$ in the dominator tree. This last task can be one in $O(1)$ time if we allow $O(|V|)$ time to pre-process the dominator tree using standard techniques.
\end{proof}

\section{Optimal $E$-coverage $E$-instrumentation}
\label{sec:reduction}

In this section we show a reduction from an instance of \prob{$E$-cov $E$-instr} to an instance of \prob{$V$-cov $V$-intr} such that an optimal solution for the latter can be transformed into an optimal solution for the former.

Given an input graph $G=(V, E)$, we construct an auxiliary graph $H$ by subdividing every edge in $E$. In other words, $H=(V \cup V_E, A)$ where $V_E = \{v_e \mid e \in E\}$ and
\begin{align*}
    A = \bigcup_{e = (u, v) \in E} \left\{(u, v_e), (v_e, v)\right\}.
\end{align*}

The first thing to note is that solving the $V$-coverage $V$-instrumentation problem in $H$ yields a solution to the problem of learning the coverage of both $V$ and $E$ by instrumenting a subset of $V$ and a subset of $E$. This is not exactly the problem we want to solve, but as we shall shortly argue, there is no additional cost in learning the coverage of $V$, and that even though we have the freedom of instrumenting vertices in $V$, we can always find an optimal solution that only instruments a subset of $E$.

\begin{lemma}
    Suppose we knew the coverage status in $H$ of every vertex in $V_E$, then we can infer the coverage status of the remaining vertices in $V$.
\end{lemma}

\begin{proof}
    For every vertex $u \in V$, the coverage status of $u$ equals the disjunction of the coverage status of edges incident on $u$.
\end{proof}

Therefore, the cost of a vertex-coverage vertex-instrumentation in $H$ is the same as the cost doing a vertex-instrumentation to learn the coverage of only $V_E$. With that out of the way, let us now reason about the vertices in $H$ that the optimal solution instruments.

\begin{lemma}
    \label{lem:never-ambiguous}
    For any $u \in V$, the corresponding node in $H$ is never ambiguous.
\end{lemma}

\begin{proof}
    This is because the vertices in $N^{in}_H(u) = \{ v_{(u,x)} : x \in N^{in}_G(u) \}$ only have $u$ as a successor, so $N^{in}_H(u) \cap B_H(u) = \emptyset$. Similarly, $N^{out}_H(u) = \{ v_{(u, x)} : x \in N^{out}_G(u) \}$ only have $u$ as a predecessor, so $N^{out}_H(u) \cap A_H(u) = \emptyset$.
\end{proof}

This means that the only way that when running algorithm \alg{optimal-instrumentation} on $H$ we instrument a vertex $u \in V$ is that the node is part of a connected component $C$ that needs to be instrumented. We use the following observation to replace each of those vertices with an equivalent edge.

\begin{lemma}
    \label{lem:alternating-component}
    Let $C$ be a strongly connected component in the inference graph of $H$ with $|C| \geq 2$. Then $C$ consists in a path alternating between vertices in $V$ and $V_E$.
\end{lemma}

\begin{proof}
    By \cref{lem:path} $C$ is a path in $H$. Notice that $H$ is a bipartite graph with shores $V$ and $V_E$. Therefore, since every path in $H$ must alternate between $V$ and $V_E$, so does must $C$.
\end{proof}

Let $S$ be an optimal vertex-coverage vertex-instrumentation for $H$. If $S \cap V = \emptyset$ then we are done as $S \subset E_V$ since this corresponds to a pure $E$-instrumentation in $G$. Otherwise, if $u \in S \cap V$ then by \cref{lem:never-ambiguous}, it must be that $u$ belongs to some strongly connected component $C$ with $|C| \geq 2$ in the inference graph of $H$. Finally, we swap out $u$ from $S$ with another vertex in $C \cap V_E$, which by \cref{lem:alternating-component} we are guaranteed to exist. 

Even easier, we can implement \alg{optimal-instrumentation} to avoid inadvertently picking a vertex from $V$ when trying to instrument each connected component. Let \alg{reduction-coverage} be the algorithm that applies this modified \alg{optimal-instrumentation} to $H$. Putting everything together we get.

\begin{theorem}
    \label{thm:E-cov-E-instr}
    \alg{reduction-coverage} return an optimal instrumentation scheme for edge-coverage edge-instrumentation in $O(|E|)$ time.
\end{theorem}

\section{Impossibility of $E$-coverage $V$-instrumentation}

\label{sec:impossible}

Let $G$ be a layered graph with the following layers $\{v_1\}$, $\{v_2, v_3\}$, $\{v_4, v_5\}$, and $\{v_6\}$, and where every layer is fully connected to the next as shown in \cref{fig:impossible}.

\begin{figure}

    \begin{center}        
        \begin{tikzpicture}[
            block/.style={circle,draw=black,fill=white, inner sep = 2pt},
            jump/.style={->},
            scale=1.2
        ]
            \draw (0, 0) node[block] (v1) {$v_1$};
            \draw (1, 1) node[block] (v2) {$v_2$};
            \draw (1, -1) node[block] (v3) {$v_3$};
            \draw (3, 1) node[block] (v4) {$v_4$};
            \draw (3, -1) node[block] (v5) {$v_5$};
            \draw (4, 0) node[block] (v6) {$v_6$};

            \draw (v1) edge[jump] (v2);
            \draw (v1) edge[jump] (v3);
            \draw (v2) edge[jump] (v4);
            \draw (v2) edge[jump] (v5);
            \draw (v3) edge[jump] (v4);
            \draw (v3) edge[jump] (v5);
            \draw (v5) edge[jump] (v6);
            \draw (v4) edge[jump] (v6);
        \end{tikzpicture}
    \end{center}

    \caption{
        \label{fig:impossible}
        An example instance where it is impossible to infer edge coverage from vertex coverage data.
    }
\end{figure}
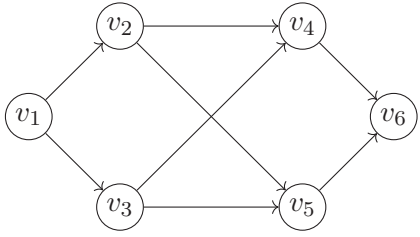

Consider two execution traces in $G$:
\[ \{ (v_1, v_2, v_4, v_6), (v_1, v_3, v_5, v_6) \} \text{ and } \{ (v_1, v_2, v_5, v_6), (v_1, v_2, v_5, v_6) \}.\] 
These two execution traces have have different edge coverage profile. However, they share the same vertex coverage profile, so it is impossible to differentiate between the two only using this data.

\section{Approximate $V$-coverage $E$-instrumentation}

\label{sec:approximation}

In this section we develop a 2-approximation algorithm for the problem of learning the coverage status of vertices using edge coverage instrumentation. Our algorithm is based on the following observation about ambiguous vertices.

\begin{lemma}
    Let $u\in V$ be an ambiguous vertex. Let $X = N^{in}(u) \cap (A(u) \cap B(u))$ and $Y = N^{out}(u) \cap (A(u) \cap (B(u)))$. Then every valid edge instrumentation scheme must instrument either $(X, u)$ or $(u, Y)$.
\end{lemma}

\begin{proof}
    First we note that both $X$ and $Y$ are non-empty by virtue of $u$ being ambiguous. Now suppose that there exists $x \in X$ and $y \in Y$ such that we do not instrument edges $(x, u)$ or $(u, y)$. Using the same logic (and the same example as in \cref{fig:ambiguous} we can conclude that there exists two execution traces that only differ in the coverage status of these two edges and $u$. Thus, if we assume that the instrumentation scheme is valid, it must be the case that we either instrument every edges in $(X, u)$ or $(u, Y)$.
\end{proof}

Our strategy is to instrument the set ($X$ or $Y$) with minimum cardinality. Notice that our choice is locally optimal in the sense that the optimal solution needs to instrument at least that many edges incident on $u$.

We use the same concept of inference graph that we developed in \cref{sec:framework}. By \cref{lem:path} we know that the only cycles present in the inference graph are induced by the edges of a directed path in the input graph. Following the same argument we used in \cref{thm:optimal}, we get that if we do note instrument a single edge incident on the vertices in the path it is not possible to infer their coverage status. On the other hand, instrumenting a single edge along the path is enough to infer status of the whole chain. Again, our choice is locally optimal.

We call this algorithm \alg{local-instrumentation}. The next theorem bounds the approximation ratio it can attain.

\begin{theorem}
    \alg{local-instrumentation} returns a valid $V$-coverage $E$-instrumentation scheme that is 2-approximate.
\end{theorem}

\begin{proof}
    The validity of the scheme follows from the observations already made. We use the local ratio technique to argue that it is a 2-approximation. For each ambiguous node $u$ we construct an edge weight function $w$ where 
    \[ w(e) = \begin{cases}
        1 & u \in e \\
        0 & o.w.
    \end{cases}
    \]
    For strongly connected components $C$ in the inference graph, we similarly defined an edge weight function $w$ where
    \[ w(e) = \begin{cases}
        1 & \exists u \in C : u \in E \\
        0 & o.w.
    \end{cases} \]
    Let $w_1, w_2, \ldots$ be the edge functions defined in this way. Furthermore, let $F$ be the edges our algorithm decides to instrument and $O$ be the edges instrumented by an optimal solution. It follows that
    \[ |S| \leq \sum_i w_i(S) \leq \sum_i w_i(O) \leq 2 |O|,\]
    where the last inequality follows from the fact that every edge $(u, v) \in O$ can contribute to the weight of the edge function defined for $u$ and for $v$.
\end{proof}

\section{Conclusion}

This paper provides a thorough theoretical study of the \prob{minimum coverage instrumentation} problem. 
Although we are able to provide definite answers to some of the variants considered, there are several problems worth studying that remain open:

\begin{itemize}
	\item What is the computational complexity of $V$-coverage $E$-instrumentation?
	
	\item In certain applications, one might be interested in learning the coverage status of a subset of
	the nodes, $S \subseteq V$. Given such a subset, we can define the $S$-coverage problem in the natural way:
	find the minimum subset of nodes (or edges) to instrument in order to be able to infer coverage of $S$.
	What is the computational complexity of the problem?
	
	\item In this paper, we focused on control-flow graphs that have a source node and a terminal node
	such that all executions start at the source and end at the terminal. However, compilers also operate with other
	types of graphs, such as \df{call graphs}, which represent calls between different functions in 
	a binary. Such a graph does not necessarily have a terminal node (for programs running continuously) and functions return control to the caller when they are done (which is not captured in our model). It would be interesting to adjust the model and the
	algorithms for such an application.
\end{itemize}

We conclude by mentioning that algorithm \alg{optimal-coverage} has been implemented in the open-source LLVM compiler project~\cite{mbc}. An extensive evaluation on real-world
benchmarks indicates that only $\approx 60\%$ of basic blocks need to be instrumented.

\bibliographystyle{abbrv}
\bibliography{refs}

\begin{thebibliography}{10}

\bibitem{asan}
\url{https://reviews.llvm.org/D17626}.

\bibitem{mbc}
Minimal block coverage in {LLVM}. \url{https://reviews.llvm.org/D124490}.

\bibitem{Agrawal94}
H.~Agrawal.
\newblock Dominators, super blocks, and program coverage.
\newblock In {\em Symposium on Principles of Programming Languages}, pages
  25--34, 1994.

\bibitem{AlstrupHLT99}
S.~Alstrup, D.~Harel, P.~W. Lauridsen, and M.~Thorup.
\newblock Dominators in linear time.
\newblock {\em {SIAM} Journal on Computing}, 28(6):2117--2132, 1999.

\bibitem{BL1994}
T.~Ball and J.~R. Larus.
\newblock Optimally profiling and tracing programs.
\newblock {\em ACM Transactions on Programming Languages and Systems (TOPLAS)},
  16(4):1319--1360, 1994.

\bibitem{BuchsbaumGKRTW08}
A.~L. Buchsbaum, L.~Georgiadis, H.~Kaplan, A.~Rogers, R.~E. Tarjan, and J.~R.
  Westbrook.
\newblock Linear-time algorithms for dominators and other path-evaluation
  problems.
\newblock {\em {SIAM} Journal on Computing}, 38(4):1533--1573, 2008.

\bibitem{BuchsbaumKRW05}
A.~L. Buchsbaum, H.~Kaplan, A.~Rogers, and J.~R. Westbrook.
\newblock \emph{Corrigendum: } {A} new, simpler linear-time dominators
  algorithm.
\newblock {\em {ACM} Trans. Program. Lang. Syst.}, 27(3):383--387, 2005.

\bibitem{CLB21}
M.~Chabbi, J.~Lin, and R.~Barik.
\newblock An experience with code-size optimization for production {iOS} mobile
  applications.
\newblock In {\em International Symposium on Code Generation and Optimization},
  pages 363--377, 2021.

\bibitem{GeorgiadisT04}
L.~Georgiadis and R.~E. Tarjan.
\newblock Finding dominators revisited: {E}xtended {A}bstract.
\newblock In {\em Annual {ACM-SIAM} Symposium on Discrete Algorithms}, pages
  869--878, 2004.

\bibitem{HMPWY22}
W.~He, J.~Mestre, S.~Pupyrev, L.~Wang, and H.~Yu.
\newblock Profile inference revisited.
\newblock {\em Symposium on Principles of Programming Languages}, 6(POPL),
  2022.

\bibitem{Knuth73}
D.~E. Knuth.
\newblock {\em The Art of Computer Programming, Volume {I:} Fundamental
  Algorithms, 2nd Edition}.
\newblock Addison-Wesley, 1973.

\bibitem{KS1973}
D.~E. Knuth and F.~R. Stevenson.
\newblock Optimal measurement points for program frequency counts.
\newblock {\em BIT Numerical Mathematics}, 13(3):313--322, 1973.

\bibitem{LeeHT22}
K.~Lee, E.~Hoag, and N.~Tillmann.
\newblock Efficient profile-guided size optimization for native mobile
  applications.
\newblock In {\em {ACM} {SIGPLAN} International Conference on Compiler
  Construction}, pages 243--253, 2022.

\bibitem{MestrePS20}
J.~Mestre, S.~Pupyrev, and S.~W. Umboh.
\newblock On the {E}xtended {TSP} problem.
\newblock In {\em International Symposium on Algorithms and Computation}, pages
  42:1--42:14, 2021.

\bibitem{Nah1973}
A.~Nahapetian.
\newblock Node flows in graphs with conservative flow.
\newblock {\em Acta Informatica}, 3(1):37--41, 1973.

\bibitem{NewellP20}
A.~Newell and S.~Pupyrev.
\newblock Improved basic block reordering.
\newblock {\em {IEEE} Transactions in Computers}, 69(12):1784--1794, 2020.

\bibitem{PH90}
K.~Pettis and R.~C. Hansen.
\newblock Profile guided code positioning.
\newblock In {\em Conference on Programming Language Design and
  Implementation}, pages 16--27, 1990.

\bibitem{Pro1982}
R.~L. Probert.
\newblock Optimal insertion of software probes in well-delimited programs.
\newblock {\em IEEE Transactions on Software Engineering}, 8(1):34--42, 1982.

\bibitem{Prosser59}
R.~T. Prosser.
\newblock Applications of boolean matrices to the analysis of flow diagrams.
\newblock In {\em Proc. of the Eastern Joint IRE-AIEE-ACM Computer Conference},
  page 133–138, 1959.

\bibitem{TikirH05}
M.~M. Tikir and J.~K. Hollingsworth.
\newblock Efficient online computation of statement coverage.
\newblock {\em J. Syst. Softw.}, 78(2):146--165, 2005.

\end{thebibliography}

\appendix

\section{Examples} \label{app:examples}

In this section we consider an example control-flow graphs where the size of the optimal solution for coverage instrumentation and frequency count instrumentation differ.

\begin{example}
    Let $G$ consist of a path $v_1, v_2, \ldots, v_k$ with self loops at every node. We claim that the size of the optimal block coverage instrumentation of $G$ is 1, whereas the size of its optimal block frequency count instrumentation is $k$.

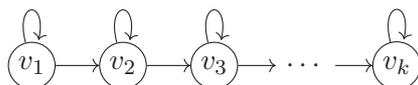
    \begin{center}        
        \begin{tikzpicture}[
            block/.style={circle,draw=black,fill=white, inner sep = 2pt},
            jump/.style={->},
            scale=1.2
        ]
            \path (0, 0) node[block] (v1) {$v_1$}
             -- ++(1, 0) node[block] (v2) {$v_2$}
             -- ++(1, 0) node[block] (v3) {$v_3$}
             -- ++(1, 0) node (dots) {$\cdots$}
             -- ++(1, 0) node[block] (vk) {$v_k$};

            \draw (v1) edge[jump] (v2);
            \draw (v1) edge[jump, loop above] (v1);
            \draw (v2) edge[jump] (v3);
            \draw (v2) edge[jump, loop above] (v2);
            \draw (v3) edge[jump] (dots);
            \draw (v3) edge[jump, loop above] (v3);
            \draw (dots) edge[jump] (vk);
            \draw (vk) edge[jump, loop above] (vk);
        \end{tikzpicture}
    \end{center}

    Indeed, it is easy to see that the optimal solution for block coverage instrumentation requires a single block as the coverage status of all blocks must be always the same. On the other hand frequency counter instrumentation requires every single block since the self-loops effectively mean that the number of time each block is executed is independent from the other blocks.
\end{example}

\begin{example}
    Let $G$ be a series parallel graph resulting from doing a serial composition of the diamond graph in \cref{ex:diamond} with itself $k$ times; namely, the graph is a sequence of diamonds as the one shown below. We claim that the optimal block coverage solution has size $2k$ whereas the optimal frequency count instrumentation has size $k+1$.

    
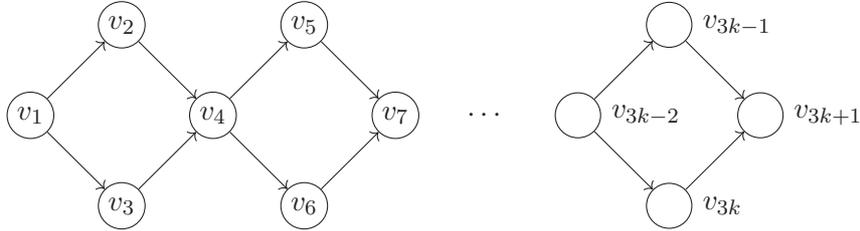
\begin{center}        
        \begin{tikzpicture}[
            block/.style={circle,draw=black,fill=white, inner sep = 2pt},
            jump/.style={->},
            scale=1.2
        ]
            \draw (0, 0) node[block] (v1) {$v_1$};
            \draw (1, 1) node[block] (v2) {$v_2$};
            \draw (1, -1) node[block] (v3) {$v_3$};
            \draw (2, 0) node[block] (v4) {$v_4$};
            \draw (3, 1) node[block] (v5) {$v_5$};
            \draw (3, -1) node[block] (v6) {$v_6$};
            \draw (4, 0) node[block] (v7) {$v_7$};
            \draw (5, 0) node (dots) {$\cdots$};
            \draw (6, 0) node[block,label=right:$v_{3k-2}$, inner sep=6pt] (v8) {};
            \draw (7, 1) node[block, label=right:$v_{3k-1}$, inner sep=6pt] (v9) {};
            \draw (7, -1) node[block,label=right:$v_{3k}$, inner sep=6pt] (v10) {};
            \draw (8, 0) node[block,label=right:$v_{3k+1}$, inner sep=6pt] (v11) {};
            \draw (v1) edge[jump] (v2);
            \draw (v1) edge[jump] (v3);
            \draw (v2) edge[jump] (v4);
            \draw (v3) edge[jump] (v4);
            \draw (v4) edge[jump] (v5);
            \draw (v4) edge[jump] (v6);
            \draw (v5) edge[jump] (v7);
            \draw (v6) edge[jump] (v7);
            \draw (v8) edge[jump] (v9);
            \draw (v8) edge[jump] (v10);
            \draw (v9) edge[jump] (v11);
            \draw (v10) edge[jump] (v11);

        \end{tikzpicture}
    \end{center}

    Indeed, all the vertices with in-degree 1 are ambiguous and there are $2k$ such vertices (2 vertices per diamond block). On the other hand, instrumenting the entry node plus a single node with in-degree 1 per diamond is enough to recover the counts of all nodes.
\end{example}

\end{document}